%% file: RR-8082.tex
 \documentclass[a4paper,10pt]{article}
\usepackage{RR,a4wide}
\usepackage[utf8x]{inputenc}
\usepackage{figlatex, subfigure}
\usepackage[american]{babel}
\graphicspath{{fig/},{graphs/},{/Users/lmarchal/travail/svn-roma/roma/utils/how_to_make_a_report/logos/}}

\usepackage[ruled,vlined]{algorithm2e}
\usepackage{color}

\usepackage{xspace, url}
\usepackage{amsthm, amsmath}
\newtheorem{definition}{Definition}
\newtheorem{theorem}{Theorem}
\newtheorem{lemma}{Lemma}
\usepackage{paralist}
\usepackage{multirow}

\usepackage{tikz}
\usetikzlibrary{trees,shapes}

\newcommand{\f}[1]{\ensuremath{\mathit{f}_{#1}}\xspace}
\newcommand{\n}[1]{\ensuremath{\mathit{n}_{#1}}\xspace}
\newcommand{\len}[1]{\ensuremath{\mathit{w}_{#1}}\xspace}
\newcommand{\W}[1]{\ensuremath{\mathit{W}_{#1}}\xspace}

\newcommand{\parent}[1]{\ensuremath{\mathit{parent}({#1})}\xspace}
\newcommand{\Children}[1]{\ensuremath{\mathit{Children}({#1})}\xspace}

\newcommand{\ParSubtrees}{\textsc{ParSubtrees}\@\xspace}
\newcommand{\ParSubtreesMath}{\textsc{ParSubtrees}} 
\newcommand{\ParSubtreesOptim}{\textsc{ParSubtreesOptim}\@\xspace}
\newcommand{\SplitSubtrees}{\textsc{SplitSubtrees}\@\xspace}
\newcommand{\ParInnerFirst}{\textsc{ParInnerFirst}\@\xspace}
\newcommand{\ParInnerFirstMemLimit}{\textsc{ParInnerFirstMemLimit}\@\xspace}

\newcommand{\ParDeepestFirst}{\textsc{ParDeepestFirst}\@\xspace}

\RRdate{October 2012}
\RRauthor{Loris Marchal \and Oliver Sinnen \and Fr\'ed\'eric Vivien}
\authorhead{L. Marchal, O. Sinnen and F. Vivien}
\RRtitle{Ordonnancement d'arbres de tâches pour minimiser mémoire et
  temps d'exécution}
\RRetitle{Scheduling tree-shaped task graphs to minimize memory and makespan}
\titlehead{Scheduling trees to minimize memory and makespan}

\RRresume{Dans ce rapport, nous nous intéressons au traitement
  d'arbres de tâches par plusieurs processeurs. Chaque arête d'un tel
  arbre représente un gros fichier d'entrée/sortie. Une tâche peut
  être traitée seulement si l'ensemble de ses fichiers d'entrée et de
  sortie peut résider en mémoire, et un fichier ne peut être retiré de
  la mémoire que lorsqu'il a été traité. De tels arbres surviennent, par
  exemple, lors de la factorisation de matrices creuses par des
  méthodes multifrontales. La quantité de mémoire nécessaire dépend de
  l'ordre de traitement des tâches. Avec un seul processeur,
  l'objectif est naturellement de minimiser la quantité de mémoire
  requise. Ce problème a déjà été étudié et des algorithmes
  polynomiaux ont été proposés.

  Nous étendons ce problème en considérant plusieurs processeurs, ce
  qui est d'un intérêt évident pour le problème de la factorisation de
  grandes matrices. Avec plusieurs processeurs se pose également le
  problème de la minimisation du temps nécessaire pour traiter
  l'arbre. Nous montrons que comme attendu, ce problème est bien plus
  compliqué que dans le cas séquentiel. Nous étudions la complexité de
  ce problème et nous fournissons des résultats d'inaproximabilité,
  même dans le cas de poids unitaires. Nous proposons plusieurs
  heuristiques pour obtenir un ordonnancement, qui se concentrent
  chacune sur un objectif différent. Nous analysons leurs performances
  par une large campagne de simulations utilisant des arbres réalistes.
}

\RRabstract{This paper investigates the execution of tree-shaped  task graphs using multiple
  processors. Each edge of such a tree represents a large IO file. A
  task can only be executed if all input and output files fit into
  memory, and a file can only be removed from memory after it has been
  consumed. Such trees arise, for instance, in the multifrontal method of
  sparse matrix factorization. The maximum amount of memory needed
  depends on the execution order of the tasks. With one processor the
  objective of the tree traversal is to minimize the required
  memory. This problem was well studied and optimal polynomial
  algorithms were proposed.

  Here, we extend the problem by considering multiple processors,
  which is of obvious interest in the application area of matrix
  factorization. With the multiple processors comes the additional
  objective to minimize the time needed to traverse the tree, i.e., to
  minimize the makespan. Not surprisingly, this problem proves to be
  much harder than the sequential one. We study the computational
  complexity of this problem and provide an inapproximability result even
  for unit weight trees. Several heuristics are proposed, each with a 
  different optimization focus, and they are analyzed in an extensive experimental
  evaluation using realistic trees.}

\RRmotcle{Ordonnancement, Contrainte Mémoire, Arbres, Optimisation bi-critère}
\RRkeyword{Scheduling, Memory-aware, Trees, Bi-objective optimization}
\RRprojet{ROMA}
\RCGrenoble

\RRNo{8082}

\begin{document}

\makeRR

\section{Introduction}

Parallel workloads are often modeled as task graphs, where nodes
represent tasks and edges represent the dependencies between
tasks. There is an abundant literature on task graph scheduling when
the objective is to minimize the total completion time, or
Makespan. However, as the size of the data to be processed is
increasing, the memory footprint of the application must be optimized
as it can have a dramatic impact on the algorithm execution time.
This is best exemplified with an application which, depending on the
way it is scheduled, will either fit in the memory, or will require
the use of swap mechanisms or out-of-core techniques.  There are very
few existing studies on the minimization of the memory footprint when
scheduling task graphs, and even fewer of them targeting parallel systems.


We consider the following memory-aware parallel scheduling problem for
rooted trees. The nodes of the tree correspond to tasks, and the edges
correspond to the dependencies among the tasks. The dependencies are
in the form of input and output files: each node takes as input
several large files, one for each of its children, and it
produces a single large file, and the different files may have
different sizes. Furthermore, the execution of any node requires its
\emph{execution} file to be present; the execution file can be seen as
the program of the task. 
We are to execute such a set of tasks on a parallel
system made of $p$ identical processing resources sharing the same
memory. The execution scheme corresponds to a schedule of the tree
where processing a node of the tree translates into reading the
associated input files and producing the output file.  How can the tree be
scheduled so as to optimize the memory usage? 

Modern computing platforms exhibit a complex memory hierarchy ranging
from caches to RAM and disks and even sometimes tape storage, with the
classical property that the smaller the memory, the quicker. Thus, to
avoid large running times, one usually wants to avoid the use of
memory devices whose IO bandwidth is below a given threshold: even if
out-of-core execution (when large data are unloaded to disks) is
possible, this requires special care when programming the application
and one usually wants to stay in the main memory (RAM). This is why in
this paper, we are interested in the question of minimizing the amount
of \emph{main memory} needed to completely process an application.

Throughout the paper, we consider \emph{in-trees} where a task can be
executed only if all its children have already been executed. (This is
absolutely equivalent to considering \emph{out-trees} as a solution
for an in-tree can be transformed into a solution for the
corresponding out-tree by just reversing the arrow of time, as outlined
in~\cite{ipdps-tree-traversal}.)
%
%
A task can be
processed only if all its
files (input, output, and execution) fit in currently available
memory. At a given time, many files may be stored in the memory, and
at most $p$ tasks may be processed by  the $p$ processors. This is
obviously possible only if all tasks and execution files fit in memory. When a
task finishes, the memory needed for its execution file and its input files
is released.
Clearly, the schedule which determines the processing times of each
task plays a key role in determining which amount of main memory is
needed for a successful execution of the whole tree.

\medskip 
The first motivation for this
work comes from numerical linear algebra. Tree workflows (assembly or
elimination trees) arise during the factorization of sparse matrices,
and the huge size of the files involved makes it absolutely necessary
to reduce the memory requirement of the factorization.
The sequential version of this problem (i.e., with $p=1$
processor) has already been studied. 
Liu~\cite{Liu86} discusses how to find a memory-minimizing traversal
when the traversal is required to correspond to a postorder traversal
of the tree. In the follow-up study~\cite{Liu87}, an exact algorithm
is shown to solve the problem, without the postorder constraint on the
traversal.  Recently, some of us~\cite{ipdps-tree-traversal} proposed
another algorithm to find a memory-optimal traversal, which proved to be
faster on existing elimination trees, although being of the same
worst-case complexity ($O(n^2)$).

The parallel version of this problem is a natural continuation
of these studies: when processing large elimination trees, it is
very meaningful to take advantage of parallel processing resources. However,
to the best of our knowledge, there exist no theoretical studies for
this problem. The key contributions of this work are:
\begin{compactitem}
\item The proof that the parallel variant of the \emph{pebble game}
  problem is NP-complete. This shows that the introduction of memory
  constraints, in the simplest cases, suffices to make the problem
  NP-hard.
\item The proof that no algorithm can simultaneously deliver a
  constant-ratio approximation for the memory minimization and for the
  makespan minimization.
\item A set of heuristics having different optimizing focus.
\item An exhaustive set of simulations using realistic tree shaped
  task graphs to assess the relative and absolute performance of these
  heuristics.
\end{compactitem}

The rest of this paper is organized as
follows. Section~\ref{sec.related} reviews related studies. The
notation and formalization of the problem are introduced in
Section~\ref{sec.model}. Complexity results are presented in
Section~\ref{sec.complexity} while Section~\ref{sec.heuristics}
proposes different heuristics to solve the problem, which are
evaluated in Section~\ref{sec.experiments}.

\section{Background and Related Work}
\label{sec.related}

\subsection{Sparse matrix factorization}

As mentioned above, determining a memory-efficient tree traversal is
very important in sparse numerical linear algebra. The elimination
tree is a graph theoretical model that represents the storage
requirements, and computational dependencies and requirements, in the
Cholesky and LU factorization of sparse matrices. In a previous study,
we have described how such trees are built, and how the multifrontal
method organizes the computations along the
tree~\cite{ipdps-tree-traversal}.  This is the context of the founding
studies of Liu~\cite{Liu86, Liu87} on memory minimization for
postorder or general tree traversals presented in the previous
section. Memory minimization is still a concern in modern multifrontal
solvers when dealing with large matrices. Among other, efforts have
been made to design dynamic schedulers that takes into account dynamic
pivoting (which impacts the weights of edges and nodes) when
scheduling elimination trees with strong memory
constraints~\cite{guermouche04ipdps}, or to consider both task and tree
parallelism with memory constraints~\cite{agulloPP12}. While these
studies try to optimize memory management in existing parallel
solvers, we aim at designing a simple model to study the fundamental
underlying scheduling problem.

\subsection{Scientific workflows}

The problem of scheduling a task graph under memory constraints also
appears in the processing of scientific workflows whose tasks require
large I/O files. Such workflows arise in many scientific fields, such
as image processing, genomics or geophysical simulations. The problem
of task graphs handling large data has been identified
in~\cite{ramakrishnan07ccgrid} which proposes some simple heuristic
solutions. Surprisingly, in the context of quantum chemistry
computations, Lam et al.~\cite{rauber11CLSS} have recently
rediscovered the algorithm published in 1987 in~\cite{Liu87}.

\subsection{Pebble game and its variants}
\label{sec.pebble}

\newcommand{\MinMemory}{\textsc{MinMemory}\xspace}
\newcommand{\MinIO}{\textsc{MinIO}\xspace}

On the more theoretical side, this work builds upon the many papers
that have addressed the pebble game and its variants.  Scheduling a
graph on one processor with the minimal amount of memory amounts to
revisiting the I/O pebble game with pebbles of arbitrary sizes that
must be loaded into main memory before \emph{firing} (executing) the
task.  The pioneering work of Sethi and Ullman~\cite{SethiUllman70}
deals with a variant of the pebble game that translates into the
simplest instance of our problem when all input/output files have
weight 1 and all execution files have weight 0. The concern
in~\cite{SethiUllman70} was to minimize the number of registers that
are needed to compute an arithmetic expression.  The problem of
determining whether a general DAG can be executed with a given number
of pebbles has been shown NP-hard by Sethi~\cite{Sethi73} if no vertex
is pebbled more than once (the general problem allowing recomputation,
that is, re-pebbling a vertex which have been pebbled before, has been
proven \textsc{Pspace} complete~\cite{gilbert80}). However, this
problem has a polynomial complexity for tree-shaped
graphs~\cite{SethiUllman70}.

To the best of our knowledge, there have been no attempts to
extend these results to parallel machines, with the objective of
minimizing both memory and total execution time. We present such an
extension in Section~\ref{sec.complexity}.


\section{Model and objectives}
\label{sec.model}

\subsection{Application model}

We consider in this paper a tree-shaped task-graph $T$ composed of $n$
nodes, or tasks, numbered from $1$ to $n$. Nodes in the tree have an
output file, an execution file (or program), and several input files
(one per child). More precisely:
\begin{itemize}
\item Each node $i$ in the tree has an execution file of size \n{i}
  and its processing on a processor takes time \len{i}.
\item Each node $i$ has an output file of size \f{i}. If $i$ is not
  the root, its output file is used as input by its parent
  $\parent{i}$; if $i$ is the root, its output file can be of size
  zero, or contain outputs to the outside world.
\item Each non-leaf node $i$ in the tree has one input file per
  child. We denote by $\Children{i}$ the set of the children of
  $i$. For each child $j \in \Children{i}$, task $j$ produces a file
  of size \f{j} for $i$. If $i$ is a leaf-node, then $\Children{i} =
  \emptyset$ and $i$ has no input file: we consider that the initial
  data of the task either reside in its execution file or are read
  from disk (or received from the outside word) during the execution
  of the task.
\end{itemize}

During the processing of a task $i$, the memory must contain its input
files, the execution file, and the output file. The memory needed for
this processing is thus:
$$
\left(\sum_{j \in \Children{i}} \f{j}\right) + \n{i} + \f{i}
$$
After $i$ has been processed, its input files and program are
discarded, while its output file is kept in memory until the
processing of its parent.


\subsection{Platform model and objectives}

In this paper, our goal is to design the simpler platform model which
allows to study memory minimization on a parallel platform. We thus
consider $p$ identical processors which share a single memory.  We do
not consider here a hard constraint on the memory, but we rather
include memory in the objectives. We thus consider multi-criteria
optimization with the following two objectives:
\begin{itemize}
\item \textbf{Makespan.} Our first objective is the classical
  makespan, or total execution time, which corresponds to the
  times-span between the beginning of the execution of the first leaf
  task and the end of the processing of the root task.
\item \textbf{Memory.} Our second objective is the amount of memory
  needed for the computation. At each time step, some files are stored
  in the memory and some task computations occur, which induces a
  memory usage. The \emph{peak memory} is the maximum usage of the
  memory over the whole schedule, which we aim at minimizing.
\end{itemize}


\section{Complexity results\\ in the Pebble Game model}
\label{sec.complexity}

Since there are two objectives, the decision version of our problem can be stated as follows.
\begin{definition}[BiObjectiveParallelTreeScheduling]
  Given a tree-shaped task graph $T$ provided with memory weights and
  task durations, $p$ processors, and two bounds $B_{C_{\max}}$ and
  $B_{\mathit{mem}}$, is there a schedule of the task graph on the
  processors whose makespan is not larger than $B_{C_{\max}}$ and
  whose peak memory is not larger than $B_{\mathit{mem}}$?
\end{definition}

This problem is obviously NP-complete. Indeed, when there are no
memory constraints ($B_{\mathit{mem}} = \infty$) and when the task
tree does not contain any inner node, that is, when all tasks are
either leaves or the root, then our problem is equivalent to
scheduling independent tasks on a parallel platform which is an
NP-complete problem as soon as tasks have different execution
times~\cite{LenstraRKBr77}. On the contrary minimizing the makespan
for a tree of same-size tasks can be solve in polynomial-time when
there are no memory constraints~\cite{Hu61}. In this section, we
consider the simplest variant of the problem. We assume that all input
files have the same size ($\forall i, \f{i} = 1$) and no extra memory
is needed for computation ($\forall i, \n{i}=0$). Furthermore, we
assume that the processing of each node takes a unit time: $\forall i,
\len{i} = 1$. We call this variant of the problem the \emph{Pebble
  Game} model since it perfectly corresponds to pebble game problems
introduced above: the weight $\f{i} = 1$ corresponds to the pebble put
on one node once it has been processed and its results is not yet
discarded. Processing a node requires to put an extra pebble on this
node and is done in unit time.

In this section, we first show that even in this simple variant, the
introduction of memory constraints (a limited number of pebbles) makes
the problem NP-hard (Section~\ref{sec:np}). Then, we show that when
trying to minimize both memory and makespan, it is not possible to get
a solution with a constant approximation ratio for both objectives
(Section~\ref{sec:inapprox}).


\subsection{NP-completeness}
\label{sec:np}

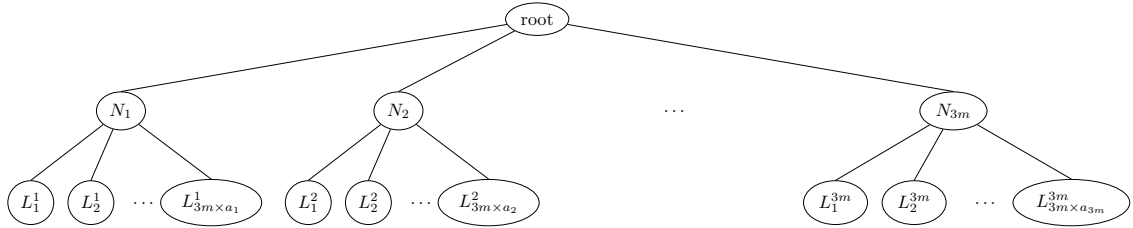
\begin{figure*}
  \centering
  \resizebox{\textwidth}{!}{
    \begin{tikzpicture}[scale=0.9, child anchor = north]
      \tikzstyle{every node}=[ellipse, draw]
      \node{root}
      [sibling distance=60mm, level distance=20mm]
      child{
        node{$N_1$}    [sibling distance=13mm]
        child{node{$L^1_1$}}
        child{node{$L^1_2$}}
        child{node[draw=none]{\ldots~~~} edge from parent [draw=none]}
        child{node{$L^1_{3m\times a_1}$}}
      }
      child{
        node{$N_2$}    [sibling distance=13mm]
        child{node{$L^2_1$}}
        child{node{$L^2_2$}}
        child{node[draw=none]{\ldots~~~} edge from parent [draw=none]}
        child{node{$L^2_{3m\times a_2}$}}
      }
      child{node[draw=none]{\ldots} edge from parent [draw=none]}
      child{
        node{$N_{3m}$}    [sibling distance=17mm]
        child{node{$L^{3m}_1$}}
        child{node{$L^{3m}_2$}}
        child{node[draw=none]{\ldots~~~} edge from parent [draw=none]}
        child{node{$L^{3m}_{3m\times a_{3m}}$}}
      }
      ;
    \end{tikzpicture}
  }
  \caption{Tree used for the NP-completeness proof}
  \label{fig:np}
\end{figure*}

\begin{theorem}
  The BiObjectiveParallelTreeScheduling problem is NP-complete in the
  Pebble Game model (i.e., with $\forall i, \f{i} = \len{i} = 1, \n{i}
  = 0$).
\end{theorem}

\begin{proof}
  First, it is straightforward to check that the problem is in NP:
  given a schedule, it is easy to compute its peak memory and
  makespan.

  To prove the problem NP-completeness, we perform a reduction from
  3-\textsc{Partition}, which is known to be
  NP-complete in the strong sense~\cite{GareyJohnson}. We consider the following instance
  $\mathcal{I}_1$ of the 3-\textsc{Partition} problem: let $a_i$ be $3m$ integers and $B$
  an integer such that $\sum a_i = m B$. We consider the variant of
  the problem, also NP-complete, where $\forall i, B/4 < a_i < B/2$.
  To solve $\mathcal{I}_1$, we need to solve the following question: does there
  exist a partition of the $a_i$'s in $m$ subsets $S_1,\ldots, S_m$,
  each containing exactly 3 elements, such that, for each $S_k$,
  $\sum_{i \in S_k} a_i = B$.  We build the following instance $\mathcal{I}_2$
  of our problem, illustrated on Figure~\ref{fig:np}. The tree
  contains a root $r$ with $3m$ children, the $N_i$'s, each one
  corresponding to a value
  $a_i$. Each node $N_i$ has $3m \times a_i$ children, which are
  leaf nodes. The question is to find a schedule of this tree on $p=3mB$
  processors, whose peak memory is not larger than $B_{\mathit{mem}} =
  3m \times B+3m$ and whose makespan is not larger than $B_{C_{\max}} =
  2m+1$.

  Assume first that there exists a solution to $\mathcal{I}_1$, i.e., that there are $m$
  subsets $S_k$ of 3 elements with $\sum_{i \in S_k} a_i = B$. In this
  case, we build the following schedule:
  \begin{itemize}
  \item At step 1, we process all the nodes $L_x^{i_1}$, $L_y^{j_1}$,
    and $L_z^{k_1}$ with $S_1 = \{a_{i_1}, a_{j_1}, a_{k_1}\}$. There are
    $3mB = p$ such nodes, and the amount of memory needed is also $3mB$.
  \item At step 2, we process the nodes $N_{i_1}$, $N_{j_1}$,
    $N_{k_1}$. The memory needed is $3mB+3$.
  \item At step $2n+1$, with $1 \leq n \leq m-1$, we process the $3mB = p$ nodes $L_x^{i_n}$, $L_y^{j_n}$,
    $L_z^{k_n}$ with $S_n = \{a_{i_n}, a_{j_n}, a_{k_n}\}$. 
    The amount of memory needed 
    is
    $3mB+3n$ (counting the memory for the output files of the $N_t$ nodes previously processed).
  \item At step $2n+2$, with $1 \leq n \leq m-1$, we process the  nodes $N_{i_n}$, $N_{j_n}$,
    $N_{k_n}$. The memory needed for this step is $3mB+3(n+1)$.
  \item At step $2m+1$, we process the root node and the memory needed
    is $3m+1$.
  \end{itemize}
  Thus, the peak memory of this schedule is $B_{\mathit{mem}}$ and its
  makespan $B_{C_{\max}}$.

  On the contrary, assume that there exists a solution to problem
  $\mathcal{I}_2$, that is, that there exists a schedule of makespan
  at most $B_{C_{\max}}=2m+1$.  Without loss of generality, we assume
  that the makespan is exactly $2m+1$. We start by proving that at any
  step of the algorithm there are at most three of the $N_i$ nodes
  that are processed. By contradiction, assume that four (or more)
  such nodes $N_{i_1}, N_{i_2}, N_{i_3}, N_{i_4}$ are processed during
  a certain step. We recall that $a_i>B/4$ so that
  $a_{i_1}+a_{i_2}+a_{i_3}+a_{i_4} > B$ and thus
  $a_{i_1}+a_{i_2}+a_{i_3}+a_{i_4} \geq B+1$. The memory needed at
  this step is thus at least $(B+1)3m$ for the children of the nodes
  $N_{i_1}$, $N_{i_2}$, $N_{i_3}$, and $N_{i_4}$ and $4$ for the nodes
  themselves, hence a total of at least $(B+1)3m+4$, which is more
  than the prescribed bound $B_{\mathit{mem}}$. Thus, at most three of
  $N_i$ nodes are processed at any step.  In the considered schedule,
  the root node is processed at step $2m+1$. Then, at step $2m$, some
  of the $N_i$ nodes are processed, and at most three of them from
  what precedes. The $a_i$'s corresponding to those nodes make the
  first subset $S_1$. Then all the nodes $L_x^j$ such that $a_j \in
  S_1$ must have been processed at the latest at step $2m-1$, and
  they occupy a memory footprint of $3m\sum_{a_j \in S_1} a_j$ at
  steps $2m-1$ and steps $2m$. Let us assume that a node $N_k$ is
  processed at step $2m-1$. For the memory bound $B_{\mathit{mem}}$
  not to be satisfied we must have $a_k + \sum_{a_j \in S_1} a_j \leq
  B$. (Otherwise, we would need a memory of at least $3m(B+1)$ for the
  involved $L_x^j$ nodes plus 1 for the node $N_k$). Therefore, node
  $N_k$ could have been processed at step $2m$. We then modify the
  schedule so as to schedule $N_k$ at step $2m$ and thus we add $k$ to
  $S_1$. We can therefore assume, without loss of generality, that no
  $N_i$ node is processed at step $2m-1$. Then, at step $2m-1$ only
  children of the $N_j$ nodes with $a_j \in S_1$ are processed, and
  all of them are. So, none of them have any memory footprint before
  step $2m-1$. We then generalize this analysis: at step $2i$, for $1
  \leq i \leq m-1$, only some $N_j$ nodes are processed and they
  define a subset $S_i$; at step $2i-1$, for $1 \leq i \leq m-1$, are
  processed exactly the nodes $L_x^k$ that are children of the nodes
  $N_j$ such that $a_j \in S_i$.


  Because of the memory constraint, each of the $m$ subsets of $a_i$'s
  built above sum to at most $B$. Since they contain all $a_i$'s,
  their sum is $mB$. Thus, each subset $S_k$ sums to $B$ and we have
  built a solution for $\mathcal{I}_1$.
\end{proof}

\subsection{Joint minimization of both objectives}

\label{sec:inapprox}

As our problem is NP-complete, it is natural to wonder whether there
exist approximation algorithms. Here, we prove that there does not
exist schedules which approximates both the minimum makespan and the
minimum memory with constant factors\footnote{This is equivalent to
  say that there is no \emph{Zenith} or \emph{simultaneous}
  approximation.}.

\begin{theorem}\label{thm.shared.inapprox}
  There is no algorithm that is both an $\alpha$-approximation for
  makespan minimization and a $\beta$-approximation for memory peak
  minimization when scheduling in-tree task graphs.
\end{theorem}

\newcommand{\cp}[2]{\ensuremath{cp^{#1}_{#2}}}
\newcommand{\nd}[2]{\ensuremath{d^{#1}_{#2}}}
\newcommand{\lf}[3]{\ensuremath{a^{#1,#2}_{#3}}}
\newcommand{\fn}[2]{\ensuremath{b^{#1}_{#2}}}
\newcommand{\cc}[2]{\ensuremath{c^{#1}_{#2}}}
\newcommand{\nwidgets}{\ensuremath{n}}

\begin{proof}
  To establish this result, we proceed by contradiction. We therefore
  assume that there is an integer $\alpha$, an integer $\beta$, and an
  algorithm $\mathcal{A}$ that processes any input tree $\mathcal{T}$
  in a time not greater than $\alpha$ times the optimal execution time
  while using a peak memory that is not greater than $\beta$ times the
  optimal peak memory.

  \textbf{The tree.} Figure~\ref{fig:tree_inapprox_shared} presents
  the tree used to derive a contradiction. This tree is made of
  $\nwidgets$ identical subtrees whose roots are the children of the
  tree root. The values of $\nwidgets$ and $\delta$ will be fixed
  later on. 
  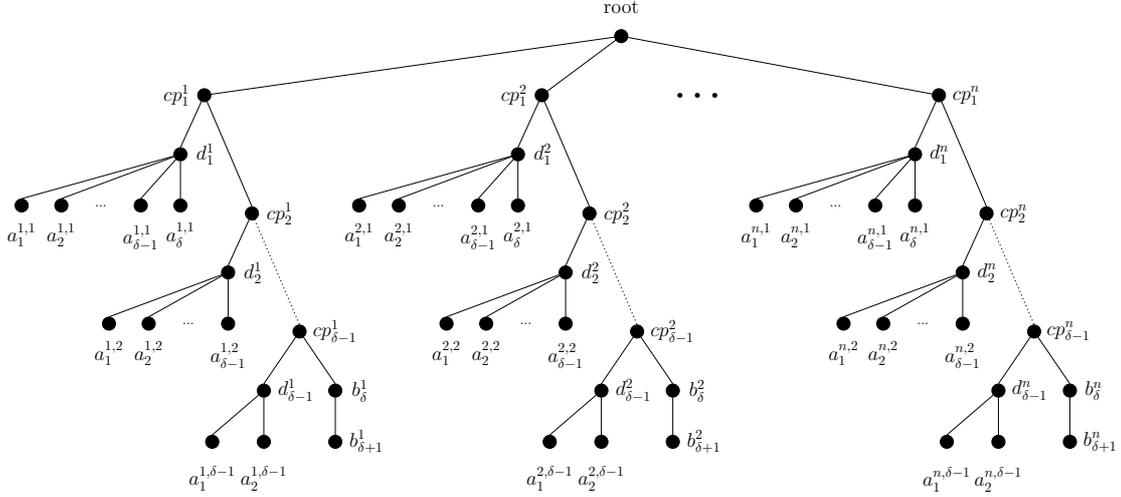
\begin{figure*}
    \centering \resizebox{\textwidth}{!}{\input{fig/fig_inapprox_constants_black}}
    \caption{Tree used for establishing Theorem~\ref{thm.shared.inapprox}.}
    \label{fig:tree_inapprox_shared}
  \end{figure*}

  \textbf{Optimal execution time.} The optimal execution time is equal
  to the length of the critical path, as we have made no hypothesis on
  the number of available processors. The critical path has a length
  of $\delta+2$, which is the length of the path from the root to any
  $\fn{i}{\delta+1}$, \lf{i}{\delta-1}{1}, or \lf{i}{\delta-1}{2}
  node, with $1 \leq i \leq \nwidgets$.

  \textbf{Optimal peak memory.} Let us consider any sequential
  execution that is optimal with regard to the peak memory
  usage. Under this execution, let \nd{i}{1} be the last processed
  node among the \nd{j}{1} nodes, $1 \leq j \leq \nwidgets$. We
  consider the step at which node \nd{i}{1} is processed. As, by
  hypothesis, all the \nd{j}{1} nodes, $1 \leq j \leq \nwidgets$
  and $j \neq i$, have already been processed, there are in memory at
  that step at
  least $\nwidgets-1$ results. The processing of \nd{i}{1} requires
  $\delta+1$ memory units as this node has $\delta$ children. Hence, a
  total memory usage of at least
  $(\nwidgets-1)+(\delta+1)=\delta+\nwidgets$ for the processing of
  \nd{i}{1}. This is obviously a lower bound on the optimal peak
  memory usage. We now show that this bound can be reached.

  We consider the following schedule:
  \begin{compactitem}

  \item Completely process first the subtree rooted at $\cp{1}{1}$,
    then the subtree rooted at \cp{2}{1}, and so on.
  \item The subtree rooted at \cp{i}{1} is processed as follows: for
    $j$ going from 1 to $\delta-1$, process the $\delta-j+1$ children
    of node \nd{j}{i}, then node \nd{j}{i}; then process nodes
    \fn{i}{\delta+1}, \fn{i}{\delta}, and nodes \cp{i}{\delta-1} to
    \cp{i}{1}.
  \end{compactitem}
  When the subtree rooted at \cp{i}{1} is processed there are in
  memory exactly $i-1$ results coming from the processing of the first
  $i-1$ subtrees. These are exactly the results of the processing of
  the nodes \cp{1}{1}, ..., \cp{i-1}{1}.

  The processing of node \nd{i}{j} requires a memory of $\delta-j+2$,
  for $1 \leq j \leq \delta-1$: this node has $\delta-j+1$ inputs and
  one output. When node \nd{i}{j} is processed the memory contains
  $j-1$ results due to the processing of the subtree rooted at
  \cp{1}{i}: the results of the processing of nodes \nd{i}{1} to
  \nd{i}{j-1}. Hence, the total memory usage when node \nd{i}{j} is
  processed is $(i-1)+(\delta-j+2)+(j-1)=i+\delta$.

  Accordingly when \fn{i}{\delta+1} is processed the memory usage is
  $(i-1)+(\delta-1)+1 = i+\delta-1$, and when \fn{i}{\delta+1} is
  processed it is $(i-1)+(\delta-1)+2 = i+\delta$. Later on, when node
  \cp{i}{j} is processed, for $1 \leq j \leq \delta-1$, the memory
  usage is $(i-1) + j + 2 = i + j + 1 \leq i + \delta$. Indeed, at
  that time, the only data in memory relative to the processing of the
  subtree rooted at \cp{i}{1} are 1) the results of the nodes
  \nd{i}{1} through \nd{i}{j}, 2) the result of \cp{i}{j+1} if $j <
  \delta-1$ or of \fn{i}{\delta} otherwise, and 3) the result of the
  processing of node \cp{i}{j}.
  
  Under this schedule, the peak memory usage during the processing of
  the subtree rooted at \cp{i}{1} is $i+\delta$. The overall peak
  memory usage of the studied schedule is then $\nwidgets+\delta$
  which is thus the optimal peak memory usage.

  \textbf{Lower bound on the peak memory usage of $\mathcal{A}$.}  The peak memory
  usage is not smaller than the average memory usage. We derive the
  desired contradiction by using the average memory usage of algorithm
  $\mathcal{A}$ as a lower bound to its peak memory usage.

  By hypothesis, algorithm $\mathcal{A}$ is $\alpha$ competitive with
  regard to makespan minimization. Therefore the processing of the
  tree by algorithm $\mathcal{A}$ should complete at the latest at
  time $\alpha (\delta+2)$.  To ensure that, the $\nwidgets$ \cp{i}{1}
  nodes, $1 \leq i \leq \nwidgets$, must all be executed at the latest
  at time $\alpha(\delta+2)-1$. Therefore, all the descendants of
  these nodes must be executed between time 0 and time
  $\alpha(\delta+2)-2$.  We now evaluate the number of these
  descendants and their memory footprints.
  
  The descendants of node \cp{i}{1} includes the two nodes
  \fn{i}{\delta} and \fn{i}{\delta+1}, the $\delta-2$ nodes \cp{i}{2}
  to \cp{i}{\delta-1}, the $\delta-1$ nodes \nd{i}{1} through
  \nd{i}{\delta-1}, and, finally, the descendants of the \nd{i}{j}
  nodes, for $1 \leq j \leq \delta-1$. As node \nd{i}{j} has
  $\delta-j+1$ descendants, the number of descendants of node
  \cp{i}{1} is:
  $$
  \displaystyle 2 + (\delta-2) + (\delta-1) + \sum_{j=1}^{\delta-1}
  (\delta-j+1) =
  \displaystyle   2\delta-1 + \left(\left(\sum_{j=1}^{\delta} j\right) -1\right)=  
  \displaystyle  \frac{\delta^2+5\delta-4}{2}
  $$
  All together, the nodes \cp{i}{1}, for $1 \leq i \leq \nwidgets$ have
  $\nwidgets \frac{\delta^2+5\delta-4}{2}$ descendants.

  We consider the memory footprint of each of these nodes between time
  step 0 and time step $\alpha(\delta+2)-2$. The result of the
  processing of each of theses nodes must be in memory for at least
  two steps in this interval, the step at which the node is processed
  and the step at which its parent node is processed, except for the
  nodes $\nd{1}{j}$, $1 \leq j \leq \nwidgets$, and $\cp{k}{2}$, for
  $1 \leq k \leq \nwidgets$, whose parents need not have been
  processed in that interval and thus need only to be present in
  memory during one time step. The overall memory footprint
  between time 0 and $\alpha(\delta+2)-2$ is then: 
  $$
  \nwidgets \left (\left (\frac{\delta^2+5\delta-4}{2}-2\right)\times
    2 + 2\times 1\right)= \nwidgets \left (\delta^2+5\delta-6\right).
  $$
  The average memory usage during that period is thus:
  $$
  \frac{\nwidgets \left (\delta^2+5\delta-6\right)}{\alpha(\delta+2)-2}.
  $$
  This is obviously a lower bound on the overall peak memory
  usage. This bound enables us to derive a lower bound $lb$ on the
  approximation ratio $\rho$ of algorithm $\mathcal{A}$ with regard to
  memory usage:
  \begin{equation*}
    \rho \geq lb = \frac{\frac{\nwidgets \left
          (\delta^2+5\delta-6\right)}{\alpha(\delta+2)-2}}{\nwidgets+\delta}
    = 
    \frac{\nwidgets \left (\delta^2+5\delta-6\right)}{(\alpha(\delta+2)-2)(\nwidgets+\delta)}.
  \end{equation*}
  We then let $\delta = \nwidgets^2$. Therefore,
  $$ lb =  \frac{\nwidgets \left
      (\nwidgets^4+5\nwidgets^2-6\right)}{(\alpha(\nwidgets^2+2)-2)(\nwidgets+\nwidgets^2)}.$$
  Then, $lb$ tends to $+\infty$ when $\nwidgets$ tends to
  infinity. There is thus a value $\nwidgets_0$ such that, for any
  value $\nwidgets \geq \nwidgets_0$, the right-hand side is greater
  than $2 \beta$. We let $\nwidgets=\nwidgets_0$ and we obtain:
  $$ lb =  \frac{\nwidgets_0 \left
      (\nwidgets_0^4+5\nwidgets_0^2-6\right)}{(\alpha(\nwidgets_0^2+2)-2)(\nwidgets_0+\nwidgets_0^2)}
  \geq 2 \beta,$$
  which contradicts the definition of $\beta$.
\end{proof}

\section{Heuristics}
\label{sec.heuristics}

Given the complexity of optimizing the makespan and memory at the same
time, we have investigated heuristics and propose three
algorithms: \ParSubtrees, \ParInnerFirst, and \ParDeepestFirst . The intention is that the proposed algorithms cover a range of use cases, where the optimization focus wanders between the makespan and the required memory. \ParSubtrees employs a memory-optimizing sequential algorithm for it subtrees, hence its focus is more on the memory side. In contrast, \ParInnerFirst and \ParSubtrees are list scheduling based algorithms, which should be stronger in the makespan objective. Nevertheless, \ParInnerFirst tries to approximate a postorder in parallel, which is good for memory in sequential. \ParDeepestFirst's focus is fully on the makespan.

The minimal memory requirement $M$ is achieved by using the optimal sequential algorithm~\cite{ipdps-tree-traversal}, i.e., using $p=1$ processor. Employing more processors cannot reduce the amount of memory required, yet the sequential algorithm is  of course only a $p$-approximation of the optimal parallel makespan $C_{max}^{*}$. 

\subsection{Heuristic \ParSubtrees}
\label{sec:parallelSubtrees}
The most natural idea to process a tree $T$ in parallel is arguably
its splitting into subtrees and their subsequent parallel processing,
each using the sequentially memory-optimal
algorithms~\cite{ipdps-tree-traversal,Liu87}. An underlying idea is
to give each processor a whole subtree in order to enable a lot
of parallelism while also limiting the increase of the peak memory usage that can be
observed when several processors work on the same subtree. Algorithm~\ref{algo.parSubtrees} outlines such an algorithm, together with the routine for splitting $T$ into subtrees given in Algorithm~\ref{algo.splitSubtrees}. The makespan obtained using \ParSubtrees is denoted by $C_{max}^{\ParSubtreesMath}$.

\LinesNumbered
\begin{algorithm}
\DontPrintSemicolon
\caption{\ParSubtrees($T$, $p$)\label{algo.parSubtrees}}
Split tree $T$ into $q$ subtrees ($q \leq p$) and remaining set of
nodes, using \SplitSubtrees($T$, $p$).\;
Concurrently process the $q$ subtrees, each using memory minimizing algorithm, e.g. \cite{ipdps-tree-traversal}.\;
Sequentially process remaining set of nodes, using memory minimizing algorithm.\;
\end{algorithm}

In this approach, $q$ subtrees of $T$, $q \leq p$, are processed in
parallel. Each of these subtrees is a maximal subtree of $T$. In other
words, each of these subtrees include all the descendants (in $T$) of its root.
The nodes not belonging to the $q$ subtrees are processed
sequentially. These are the nodes where the $q$ subtrees merge, the
nodes included in subtrees that where produced in excess (if more than
$p$ subtrees where created), and the ancestors of these nodes. 
An alternative approach, as discussed below, is to process all subtrees in parallel, assigning more than one subtree to each processor, but Algorithm~\ref{algo.parSubtrees} allows us to find a \textit{makespan}-optimal splitting into subtrees, established shortly in Lemma~\ref{lem:SplitSubtreesIsOptimal}.  

As \len{i} is the computation weight of node $i$, \W{i} denotes the total computation weight (i.e., sum of weights) of all nodes in the subtree rooted in $i$, including $i$.
\SplitSubtrees uses a node priority queue $PQ$ in which the nodes are
sorted by non-increasing \W{i}, and ties are broken according to
non-increasing \len{i}. $head(PQ)$ returns the first node of $PQ$,
while $popHead(PQ)$ also removes it. $PQ[i]$ denotes the $i$-th
element in the queue.

\SplitSubtrees starts with the root and continues splitting the
largest subtree (in terms of \W{}) until this subtree is a leaf node
($\W{head(PQ)} = \len{head(PQ)}$). The execution time of Step 2
of \ParSubtrees is that of the largest of the $q$ subtrees, hence
\W{head(PQ)} of the splitting. Splitting subtrees that are smaller
than the largest leaf ($\W{j}<\max_{i \in T}\len{i}$) cannot decrease
the parallel time, but only increase the sequential time. More
generally, given any splitting $s$ of $T$ into subtrees, the best
execution time for $s$ with \ParSubtrees is achieved by choosing the
$p$ largest subtrees for the parallel Step 2. This can be easily
derived, as swapping a large tree included in the sequential part with
a smaller tree included in the parallel part cannot increase the total execution time. 


\begin{algorithm}
\DontPrintSemicolon
\caption{\SplitSubtrees($T$, $p$)\label{algo.splitSubtrees}}
  Compute weights $\W{i}, \forall i \in T$\;
  $PQ \leftarrow root$\;
  $seqSet \leftarrow \emptyset$\;
  $Cost(0)=\W{root}$\;
  $s \leftarrow 1$ \tcc*{splitting rank}
  \While {$\W{head(PQ)} > \len{head(PQ)}$} {
    $node \leftarrow popHead(PQ)$\;
    $seqSet \leftarrow seqSet \cup node$\;
    $PQ \leftarrow  \Children{node}$\;
    $C_{max}^{\ParSubtreesMath}(s) = \W{head(PQ)} +
    \quad ~ ~ ~ ~ ~ ~ ~ ~ ~   \quad  \sum_{i \in seqSet}\len{i} + \sum_{i=PQ[p+1]}^{|PQ|}\W{i}$\;
    $s \leftarrow s+1$\;
  }
  Select splitting $x$ with $C_{max}^{\ParSubtreesMath}(x)=\min_{t=0}^{s-1} C_{max}^{\ParSubtreesMath}(t)$\; 
\end{algorithm}

\begin{lemma}\label{lem:SplitSubtreesIsOptimal}
 \SplitSubtrees returns a splitting of $T$ into subtrees that results in the \textit{makespan}-optimal processing of $T$ with \ParSubtrees. 
\end{lemma}

\begin{proof}
 The proof is by contradiction. Let $S$ be the splitting into subtrees
 selected by \SplitSubtrees. Assume now that there is a different splitting $S_{opt}$ which results in a shorter processing with \ParSubtrees.
 
 Let $r$ be the root node of a heaviest subtree in $S_{opt}$. Let $t$
 be the first step in \SplitSubtrees where a node, say $r_t$, of
 weight \W{r} is the head of $PQ$ at the end of the step  ($r_t$ is
 not necessarily equal to $r$, as there can be more than one subtree
 of weight \W{r}). There is always such a step $t$, because all
 subtrees are split by \SplitSubtrees until at least one of the
 largest trees is a leaf node. By definition of $r$, there cannot be
 any leaf node heavier than \W{r}. The cost of the solution of step
 $t$ is $C_{max}^{\ParSubtreesMath}(t)=\W{r}+Seq(t)$, hence parallel
 time plus sequential time, denoted by $Seq(t)$. $Seq(t)$ is the total
 weight of the sequential set $seqSet$ plus the total weight of the
 surplus subtrees (that is, of all the subtrees except the $p$ ones of
 largest weights). The cost of $S_{opt}$ is $C_{max}^{*}=\W{r}+Seq(S_{opt})$, given that $r$ is the root of a heaviest subtree of $S_{opt}$ by definition.

 The splitting at step $t$ (and any other splitting considered by
 \SplitSubtrees) cannot be identical to $S_{opt}$, otherwise
 \SplitSubtrees would have selected that splitting. All subtrees that
 were split in \SplitSubtrees before step $t$ were strictly heavier
 than \W{r}. Thus, there cannot exist any subtree in $S_{opt}$, whose subtrees are part of the splitting at step $t$. Hence for every subtree $T_j$
 in the splitting at step $t$ the following property holds: either
 $T_j$ is part of $S_{opt}$ or a splitting of $T_j$ into subtrees is
 part of $S_{opt}$. It directly follows that $Seq(t) \leq
 Seq(S_{opt})$, because every splitting of a tree into subtrees
 increases the sequential time by at least the root's weight. As the
 parallel time is identical for $t$ and $S_{opt}$, namely \W{r}, it
 follows that $C_{max}^{\ParSubtreesMath}(t) \leq C_{max}^{*}$, which is a contradiction to $S_{opt}$'s shorter processing time. 
\end{proof}

\paragraph{Complexity} We first analyse the complexity of \SplitSubtrees. Computing the weights \W{i} costs $O(n)$. 
Each insertion into $PQ$ costs $O(\log(n))$ and calculating $C_{max}^{\ParSubtreesMath}(s)$ in each step costs $O(p)$. 
Given that there are $O(n)$ steps, \SplitSubtrees's complexity is $O(n(\log(n)+p))$. 
The complexity of the sequential traversal algorithms used in Steps 2 and 3 of \ParSubtrees cost at most $O(n^2)$, e.g., \cite{ipdps-tree-traversal,Liu87}, or $O(n\log(n))$ if the optimal postorder suffices. Thus the total complexity of \ParSubtrees is $O(n^2)$ or $O(n\log(n))$, depending on the chosen sequential algorithm.    

\ParSubtrees has the following guarantees for the memory requirement and makespan.

\paragraph{Memory} \ParSubtrees is a $(p+1)$-approximation algorithm for peak
memory minimization.
During the parallel part of \ParSubtrees the total memory used is less than $p$ times the memory for the complete sequential execution ($M_{seq}$), $M_{p} \leq p \cdot M_{seq}$. 
This is because each of the $p$ processors
executes a maximal subtree and that the processing of any subtree uses,
obviously, less memory (if done optimally) than the processing of the
whole tree.
During the sequential part of \ParSubtrees the memory is bounded by $M_{s} \leq M_{seq} + p \cdot \max_{i \in T}\f{i} \leq (p+1)M_{seq}$, where the second term is for the output files produced by the up to $p$ subtrees processed in parallel. Hence, in total: $M \leq (p+1) M_{seq}$  

\paragraph{Makespan} \ParSubtrees delivers a $p$-approximation
algorithm for
makespan minimization. In other
words, the makespan achieved by \ParSubtrees can be up to $p$ times
worse than the optimal makespan and thus may be not faster than the
sequential execution.
This can be derived readily with a tree of height 1 and $p \cdot k$
leaves (a fork) and $\len{i}=1, \forall i \in T$, where $k$ is a large
integer (this tree is depicted on Figure~\ref{fig:fork}). 
The optimal makespan for such a tree is $C_{max}^{*}=kp/p+1=k+1$. 
With \ParSubtrees the makespan is $C_{max}= 1 + (1+pk-p)=p(k-1)+2$. When $k$ tends to $+\infty$ the ratio between the makespans tends to $p$.

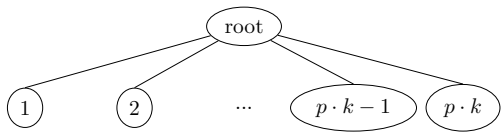
\begin{figure}
  \centering
  \scalebox{0.8}{
\begin{tikzpicture}[scale=0.9, child anchor = north]
\tikzstyle{every node}=[ellipse, draw]
  \node{root}
  [sibling distance=20mm, level distance=15mm]
  child{node{1}}
  child{node{2}}
  child{node[draw=none]{...} edge from parent [draw=none]}
  child{node{$p \cdot k-1$}}
  child{node{$p \cdot k$}}
  ;
\end{tikzpicture}
}
  \caption{\ParSubtrees is at best a $p$-approximation for the makespan.}
  \label{fig:fork}
\end{figure}

Given the just observed worst case for the makespan, a makespan
optimization for \ParSubtrees is to allocate all produced subtrees to
the $p$ processors instead of only $p$. This can be done by ordering
the subtrees by non-increasing total weight and allocating each
subtree in turn to the processor with the lowest total weight. Each of
the parallel processors executes its subtrees sequentially. This
optimized form of the algorithm shall be named \ParSubtreesOptim. Note
that this optimization should improve the makespan, but it will likely worsen the peak memory usage.


\subsection{Heuristic \ParInnerFirst}

\ParSubtrees is a high level algorithm employing sequential memory-optimized algorithms. An alternative is to design algorithms that directly work on the tree in parallel and we present two such algorithms. 
From the sequential case it is known that a \textit{postorder} traversal, while not optimal for all instances, provides good results~\cite{ipdps-tree-traversal}.
Our intention is to extend the principle of postorder traversal to the parallel processing. To do so we establish the following rules.

\noindent \textit{Parallel Postorder:}
\begin{enumerate}
 \item If an inner node (i.e., a non-leaf node) is ready to be
   processed (i.e., its input files are all in memory) then execute it. 
 \item Otherwise, select and process the leaf node that is closest (in terms of edges to be traversed) to the previously selected leaf. 
\end{enumerate}
These rules do not correspond to the usual formulation of postorder
but, when applied using a single processor, they  give rise to a
postorder traversal of the tree. Due to the concurrent processing of
nodes with $p$ processors, the resulting order will not be a perfect
postorder, but hopefully a close approximation.  

With the careful formulation of the parallel postorder we are able to base the heuristic on an event-based list scheduling algorithm \cite{Hwang1989:spg}.
Algorithm~\ref{algo.listScheduling} outlines a generic list
scheduling, driven by node finish time events. At each event at least
one node has finished so at least one processor is available for
processing nodes. Each available processor is given the respective head node of the priority queue.

\begin{algorithm}
\DontPrintSemicolon
\caption{List scheduling($T$, $p$, $O$) \label{algo.listScheduling}}
  Insert leaves in $PQ$, ordered as in $O$\;
  $eventSet \leftarrow \{0\}$   \tcc*{ascending order}
  \While(\tcc*[f]{event:node finishes}){$eventSet \neq \emptyset$}{
    $popHead(eventSet)$\;
    Insert new ready nodes in $PQ$ \tcc*{available parents of nodes
      completed at event}
    $P_a \leftarrow$ available processors\;
    \While{$P_a \neq \emptyset$ and $PQ \neq \emptyset$}{
      $proc \leftarrow popHead(P_a)$; $node \leftarrow popHead(PQ)$\;
      Assign $node$ to $proc$\;
      $eventSet~\leftarrow~eventSet \  \cup \ \ \
   ~ \ \ \ \ \ \ \ \ \ \ \ \ \ \ \ \ \ \ \ \ \   finishTime(node)$\;
    }
   }
\end{algorithm}

The order in which nodes are processed in Algorithm~\ref{algo.listScheduling} is determined by two aspects: i) the node order $O$ given as input; and ii) the ordering established by the priority queue $PQ$. 

For our proposed parallel postorder algorithm, called \ParInnerFirst, the priority queue uses the following ordering: 1) inner nodes, ordered by non-increasing depth; 2) leaf nodes as ordered in the input order $O$. To achieve a parallel postorder, the node ordering $O$ needs to be a  sequential postorder. It makes heuristic sense that this postorder is an optimal sequential postorder, so that memory consumption can be minimized ~\cite{Liu86}.  

\paragraph{Complexity}
The complexity of \ParInnerFirst is that of determining the input order $O$ and that of the list scheduling. Computing the optimal sequential postorder is $O(n\log{n})$~\cite{Liu86}. In the list scheduling algorithm there are  $O(n)$ events and $n$ nodes are inserted and retrieved from $PQ$. An insertion into $PQ$ is $O(\log{n})$, so the list scheduling complexity is $O(n\log{n})$. Hence, the total complexity is also $O(n\log{n})$. 

In the following we study the memory requirement and makespan of \ParInnerFirst. 

\paragraph{Memory}
There is no limit on the required memory compared to the optimal
sequential memory $M_{seq}$. This is derived considering the tree
in Figure~\ref{fig:NoMemBoundInnerFirst}. All output files have size
1 and the execution files have size 0 ($\f{i}=1,\n{i}=0$ for any node
$i$ of $T$). When optimally processing with $p=1$, we process the
leaves in a deepest first order. The resulting optimal memory
requirement is $M_{seq}=p+1$, reached when processing a join node.
With $p$ processors all leaves have been processed at the time the first join node ($k-1$) can be executed. (The longest chain has length $2k$.) At that time there are $(k-1)\cdot(p-1)+1$ files in memory. When $k$ tends to $+\infty$ the ratio between the memory requirements also tends to $+\infty$.     

    \begin{figure}
    \centering
    \begin{tikzpicture}[scale=0.5,
     level/.style={level distance=12mm, sibling distance=7mm},
     vertex/.style={circle,solid,draw},
     ghost/.style={circle}]

      \makeatletter
      \let\mypgfutil@ifnch\pgfutil@ifnch
      \def\pgfutil@ifnch{%
      \let\x@next\@empty
      \ifx\pgfutil@let@token\children\let\pgfutil@let@token c\let\x@next\expandafter\fi
      \ifx\pgfutil@let@token\chain\let\pgfutil@let@token n\let\x@next\expandafter\fi
      \x@next\mypgfutil@ifnch}
      \makeatother
      
      \newcommand{\chain}{%
        node[vertex,label=right:$k$] {} child{ node[vertex] {} edge from parent[solid] child{node[vertex] {} edge from parent[dotted] child{node[vertex,label=right:$2k$] {} edge from parent[solid]}}}%
      }%

      \newcommand{\children}{%
        child{node[vertex,label=below:1] {} edge from parent[solid]}
        child{node[ghost] {$\cdots$} edge from parent[draw=none]}
	child{node[vertex] {} edge from parent[solid]} 
        child{node[ghost] {$\cdots$} edge from parent[draw=none]} 
        child{node[vertex,label=below:$p-1$] {} edge from parent[solid]} 
        child{node[ghost] {} edge from parent[draw=none]}%
      }%

      \newcommand{\childrennolabel}{%
        child{node[vertex] {} edge from parent[solid]}
        child{node[ghost] {$\cdots$} edge from parent[draw=none]}
	child{node[vertex] {} edge from parent[solid]} 
        child{node[ghost] {$\cdots$} edge from parent[draw=none]} 
        child{node[vertex] {} edge from parent[solid]} 
        child{node[ghost] {} edge from parent[draw=none]}%
      }%
     
     \footnotesize 
     \node[vertex,label=right:1] (root) {}  
       \children
       child{node[vertex] {}
         child{node[vertex,right=45] {} edge from parent[dotted]
	   \childrennolabel
	   child{node[vertex,label=right:$k-1$] {} edge from parent[solid]
             \children
             child{\chain}
           }
         }
       }
     ;

    \end{tikzpicture}
      \caption{No memory bound for \ParInnerFirst.}\label{fig:NoMemBoundInnerFirst}
    \end{figure}
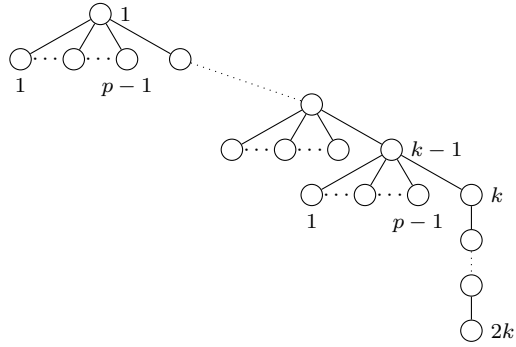

%

\paragraph{Makespan} \ParInnerFirst schedule is a
$(2-\frac{1}{p})$-approximation algorithm for makespan minimization
because \ParInnerFirst is a list scheduling algorithm~\cite{GrahamList}.

\subsection{Heuristic \ParDeepestFirst}
The previous heuristic \ParInnerFirst is motivated by good memory results for sequential postorder. Going the opposite direction, a heuristic objective can be the minimization of the makespan. For trees, all inner nodes depend on the leaf nodes, so it makes heuristic sense to try to process the deepest nodes first to reduce any possible waiting time. For the parallel processing of the tree, the most meaningful definition of the depth of a node $i$ is the \len{}-weighted length of the path from $i$ to the root of the tree. This path length includes the \len{i}. The deepest node is the first node of the critical path of the tree.

\ParDeepestFirst is our proposed algorithm that does this. Due to the general nature of the list scheduling presented in Algorithm~\ref{algo.listScheduling}, we can implement \ParDeepestFirst with it. To achieve the deepest first processing the priority queue $PQ$ orders the nodes as follows: 
1) deepest nodes first (in terms of \len{}-weighted path length to root); 2) inner nodes before leaf nodes; 3) leaf
nodes are ordered in the input order $O$. Note that the leaf order is
only relevant for leaves of the same depth. This order should
nevertheless be ``reasonable'', i.e., it should not alternate between leaves from different parents, which would be bad for the memory consumption. Such an order is again easily achieved when $O$ is a sequential postorder.    


\paragraph{Complexity}
The complexity is the same as for \ParInnerFirst, namely $O(n\log{n})$. See \ParInnerFirst's complexity analysis. 

Now we study the memory requirement and the makespan of \ParDeepestFirst. 

\paragraph{Memory}
The required memory of \ParDeepestFirst is unbounded compared to the
optimal sequential memory $M_{seq}$. Consider the tree in
Figure~\ref{fig.combWithChains} with many long chains, assuming the
Pebble Game model (i.e., $\f{i}=1$, $\n{i}=0$, and $\len{i}=1$ for any
node $i$ of $T$). The optimal sequential memory requirement is 3.
The memory usage of \ParDeepestFirst will be proportional to
the number of leaves, because they are all at the same depth, the
deepest one. As we can build a tree like the one of
Figure~\ref{fig.combWithChains} for any predefined number of chains,
the ratio between the memory required by \ParDeepestFirst and the
optimal one is unbounded.

    \begin{figure}
    \centering
    \begin{tikzpicture}[scale=0.5]
     \tikzstyle{level 1}=[level distance=12mm, sibling distance=20mm]
     \tikzstyle{vertex}=[circle,draw] 

     \node[vertex]{}
       child{node[vertex]{}
         child{node[vertex]{}
           child{node[vertex]{} edge from parent[dotted,thick]
             child[solid,thin]{node[vertex]{}
               child{node[vertex]{}
                 child{node[vertex]{}}
               }
             }
           }
         } 
       }
       child{node[vertex]{}
         child{node[vertex]{}
           child{node[vertex]{} edge from parent[dotted,thick]
             child[solid,thin]{node[vertex]{}
               child{node[vertex]{}
                 child{node[vertex]{}}
               }
             }
           }
         } 
         child{node[draw=none]{} edge from parent[dotted,thick]   
           child[]{node[draw=none]{} edge from parent[draw=none]
             child{node[draw=none]{} edge from parent[draw=none]
               child{node[draw=none]{} edge from parent[draw=none]
                 child{node[draw=none]{} edge from parent[draw=none]}
               }
             }
           } 
           child[thin,solid]{node[vertex]{} edge from parent[dotted, thick]
             child[thin,solid]{node[vertex]{}
               child{node[vertex]{}
                 child{node[vertex]{}}
               } 
             }
             child[thin,solid]{node[vertex]{}
               child{node[vertex]{}
                 child{node[vertex]{}}
               }
               child{node[vertex]{}
                 child{node[vertex]{}}
                 child{node[vertex]{}}
               }
             }
           }
         }
       }
     ;

    \end{tikzpicture}
      \caption{Tree with long chains.} \label{fig.combWithChains}
    \end{figure}
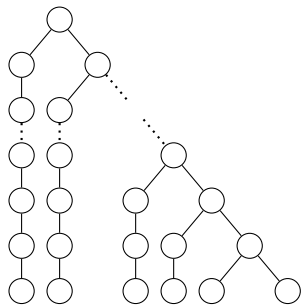

\paragraph{Makespan} \ParDeepestFirst schedule is a
$(2-\frac{1}{p})$-approximation algorithm for makespan minimization
because \ParDeepestFirst is, like \ParInnerFirst, a list scheduling algorithm~\cite{GrahamList}.


\section{Experimental validation}
\label{sec.experiments}

In this section, we experimentally compare the heuristics proposed in
the previous section, and we compare their performance to lower bounds.

\subsection{Setup}

All heuristics have been implemented in C. Special care has been
devoted to the implementation to avoid complexity issues. Especially,
priority queues have been implemented using binary heap to allow for
$O(\log n)$ insertion and minimum extraction\footnote{
The code and the data sets 
are available
online at~\url{http://graal.ens-lyon.fr/~lmarchal/scheduling-trees/}}. 

Instead of implementing an intricate algorithm with $O(n^2)$
complexity such as Liu's algorithm~\cite{Liu87} to obtain minimum
sequential memory, we have chosen to estimate this minimum memory
using the optimal post-order traversal. We have shown
in~\cite{ipdps-tree-traversal} that this traversal was optimal in
95.8\% of the tested cases, with an average increase of 1\% with
respect to the optimal. This justifies
this choice. Since the reference sequential task-graph traversal serves as a basis for ordering
nodes in a number of our heuristics, a large complexity would be
prohibitive for this first step.

\subsection{Data set}

The data set contains assembly trees of a set of sparse matrices
obtained from the University of Florida Sparse Matrix Collection
(\url{http://www.cise.ufl.edu/research/sparse/ matrices/}). The
chosen matrices satisfy the following assertions: not binary, not
corresponding to a graph, square, having a symmetric pattern, a number
of rows between 20,000 and 2,000,000, a number of nonzeros per row
at least equal to 2.5, and a number of nonzeros per row at most
equal to 5,000,000; and each chosen matrix 
has the largest number of
nonzeros among the matrices in its group satisfying the
previous assertions. At the time of testing there were 76 matrices
satisfying these properties.  We first order the matrices using
MeTiS~\cite{kaku:98:metis} (through the MeshPart toolbox~\cite{gimt:98})
and {\tt amd} (available in Matlab), and then build the corresponding
elimination trees using the {\tt symbfact} routine of Matlab.  We also
perform a relaxed node amalgamation on these elimination trees to create
assembly trees. We have created a large set of instances by allowing
1, 2, 4, and 16 (if more than $1.6 \times 10^5$ nodes) relaxed amalgamations per
node.
At the end we compute memory weights and processing times to
accurately simulate the matrix factorization: we compute the memory
weight \n{i} of a node as $\eta^2 +2\eta(\mu-1)$, where $\eta$ is the
number of nodes amalgamated, and $\mu$ is the number of nonzeros in
the column of the Cholesky factor of the matrix which is associated
with the highest node (in the starting elimination tree); the
processing cost \len{i} of a node is defined as $2/3\eta^3 +
\eta^2(\mu-1) + \eta (\mu-1)^2$ (these terms corresponds to one
gaussian elimination, two multiplications of a triangular
$\eta\times\eta$ matrix with a $\eta\times(\mu-1)$ matrix, and one
multiplication of a $(\mu-1)\times \eta$ matrix with a $\eta \times
(\mu-1)$ matrix).  The memory weights \f{i} of edges are computed as
$(\mu-1)^2$.

The resulting 608 trees contains from 2,000 to 1,000,000 nodes. Their
depth ranges from 12 to 70,000 and their maximum degree ranges from 2
to 175,000. Each heuristic is tested on each tree using $p=2$, 4, 8,
16, and 32 processors. Then the memory and makespan of the resulting
schedules are evaluated by simulating a parallel execution.

\subsection{Results}

\begin{table*}[htbp]
  \centering
  \resizebox{\textwidth}{!}{%
  \begin{tabular}[tabular]{|c||c|c|c||c|c|c|}
    \hline
\multirow{2}{*}{Heuristic}  & \multirow{2}{*}{Best memory}  & Within 5\% of & Avg. deviation from  & \multirow{2}{*}{Best makespan}  & Within 5\% of & Avg. deviation\\ 
    &  &  best memory & optimal (seq.) memory  &  & best makespan & from best makespan\\ 
    \hline
\ParSubtrees  & 81.1 \% & 85.2 \% & 133.0 \%     & 0.2 \%  & 14.2 \%  & 34.7 \% \\
\ParSubtreesOptim  & 49.9 \% & 65.6 \% & 144.8 \%     & 1.1 \%  & 19.1 \%  & 28.5 \% \\
\ParInnerFirst  & 19.1 \% & 26.2 \% & 276.5 \%     & 37.2 \%  & 82.4 \%  & 2.6 \% \\
\ParDeepestFirst  & 3.0 \% & 9.6 \% & 325.8 \%     & 95.7 \%  & 99.9 \%  & 0.0 \% \\
    \hline
  \end{tabular}}
  \caption{Proportions of scenarii when heuristics reach best (or
    close to best) performance, and average deviations from optimal
    memory and best achieved makespan.}
  \label{tab.cmp}
\end{table*}

\begin{figure*}[htbp]
  \centering
  \includegraphics[width=0.7\linewidth]{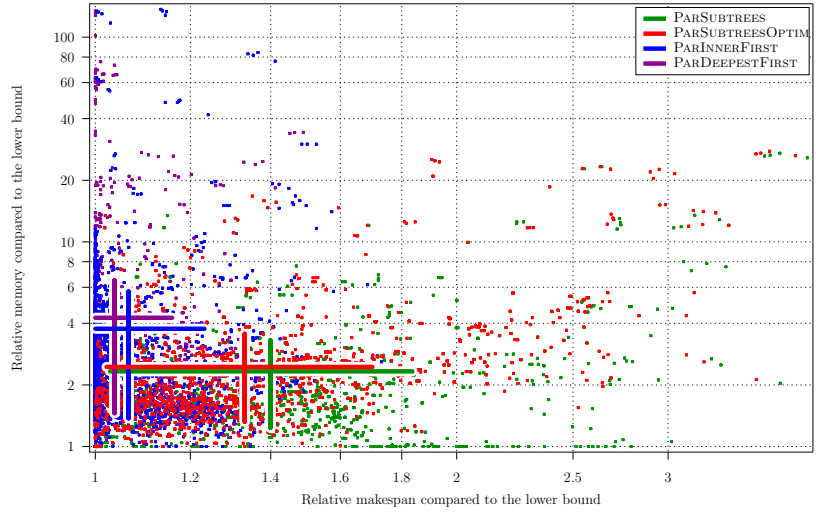}
  \caption{Comparison to lower bounds.}
  \label{fig.cmp.lb}
\end{figure*}

The comparison of the heuristics is summarized in
Table~\ref{tab.cmp}. It shows that \ParSubtrees and \ParSubtreesOptim
are the best heuristics for memory minimization. On average they use
less than 2.5 times the amount of memory required by the best
sequential postorder (whose memory usage is very close to the optimal sequential memory as
noted above), when \ParInnerFirst and \ParDeepestFirst need
respectively 3.7 and 5.2 times this amount of memory. \ParInnerFirst
and \ParDeepestFirst perform best for makespan minimization, having
makespans very close on average to the best achieved ones. As the
scheduling problem, without memory constraints, is already NP-hard, we
do not know what the optimal makespan is. We have seen however
that \ParInnerFirst and \ParDeepestFirst are 2-approximation
algorithms for the makespan. Furthermore, given the critical path
oriented node ordering, we can expect that \ParDeepestFirst's makespan
is close to optimal. \ParInnerFirst outperforms \ParInnerFirst for makespan minimization, at the cost of a
noticeable increase in memory. 
\ParSubtrees and \ParSubtreesOptim may be better trade-offs, since
their average deviation from best makespan is under 35\%.



\begin{figure*}[htbp]
  \centering
  \includegraphics[width=0.7\linewidth]{errorbars-parallel-subtrees.fig}
  \caption{Comparison to \ParSubtrees.}
  \label{fig.cmp.parsubtrees}
\end{figure*}

\begin{figure*}[htbp]
  \centering
  \includegraphics[width=0.7\linewidth]{errorbars-inner-first.fig}
  \caption{Comparison to \ParInnerFirst.}
  \label{fig.cmp.innerfirst}
\end{figure*}

Figures~\ref{fig.cmp.lb},~\ref{fig.cmp.parsubtrees},
and~\ref{fig.cmp.innerfirst} provide complete results of the
simulations. In each figure, a point represent one scenario (one
heuristic on one tree with a given number of processors). To better
visualize the distribution, we also plot a ``cross'' for each
heuristic: the center of this cross is the average performance, while
the branches represent the scope of each objective between the 10th
and the 90th percentile of the distribution.

On Figure~\ref{fig.cmp.lb}, we plot the results of all simulations
compared to some estimations of the lower bounds. The lower bound for
memory minimization is the memory usage of the best sequential
postorder, which is known to be very close to the optimal sequential
traversal. The lower bound for the makespan is the maximum between the
total processing time of the tree divided by the number of processors,
and the maximum weighted critical path. This
figure exhibits the same trends for average values as noted in
Table~\ref{tab.cmp}. When the maximum deviation from the lower bound
on the makespan is around 4, the ratio of the parallel memory usage to
the optimal sequential one can be far larger, as it is larger than
100 for the extreme cases.

In the following figures, the results of the heuristics 
 is normalized by the results of \ParSubtrees
(Figure~\ref{fig.cmp.parsubtrees}) or \ParInnerFirst
(Figure~\ref{fig.cmp.innerfirst}). 
As expected, \ParSubtreesOptim gives results
close to those of \ParSubtrees, with better makespans but slightly worse
memory usage. \ParDeepestFirst always use more memory than \ParInnerFirst,
while having comparable makespans. In most cases, \ParInnerFirst gives slightly
better makespan than \ParSubtrees, but uses more memory.




  






\section{Conclusion}

In this study we have shown that the parallel version of the pebble
game on trees is NP-complete, hence stressing the negative impact of the memory
constraints on the complexity of the problem. More importantly, we
have shown that there does not exist any algorithm that is
simultaneously an approximation algorithm for both makespan
minimization and peak memory usage minimization when scheduling
tree-shaped task graphs. We have thus designed heuristics for this
problem. We have assess their performance using real task graphs
arising from sparse matrices computation. These simulations showed
that two of the heuristics, \ParSubtrees and \ParSubtreesOptim, only
needed, for their parallel executions, and on average, 2.5 times the
sequential memory, while achieving makespans that were less than 35\%
larger than best achieved ones. These heuristics appear thus to
deliver interesting trade-offs between memory usage and execution
times.  In the future work, we will consider designing scheduling
algorithms that take as input a cap on the memory usage.

\section*{Acknowledgement}

We thank Bora Uçar for his help in creating and
managing the data set used in the experiments.
We gratefully acknowledge that this work is partially supported by the
Marsden Fund Council from Government funding, Grant 9073-3624767,
administered by the Royal Society of New Zealand.

\bibliographystyle{plain}
\bibliography{biblio}

\end{document}

%% file: fig/fig_inapprox_constants_black.tex
\begin{tikzpicture}[child anchor = north, label distance = -2pt]
  \tikzstyle{every node} = [circle, fill=black, draw]
  \tikzstyle{every label} = [font=\Large,fill=none,draw=none]
  \node (root) [label=above:root]{}
  --   (-105mm,-15mm) 
     node[label=left:\cp{1}{1}]{}
    [sibling distance=12mm, level distance=15mm]
    child{
      node[label=right:\nd{1}{1}]{} 
      [grow via three points={one child at (0,-1.3) and two children at (0,-1.3) and (-1,-1.3)}]
      child{node[label=below:\lf{1}{1}{\delta}]{}}
      child{node[label=below:\lf{1}{1}{\delta-1}]{}}
      child{node[draw=none, fill=white]{...} edge from parent [draw=none]}
      child{node[label=below:\lf{1}{1}{2}]{}}
      child{node[label=below:\lf{1}{1}{1}]{}}
    }
    child[level distance=30mm,sibling distance=24mm]{ 
      node[label=right:\cp{1}{2}]{}
      [sibling distance=12mm, level distance=15mm]
      child{
        node[label=right:\nd{1}{2}]{} 
        [grow via three points={one child at (0,-1.3) and two children at (0,-1.3) and (-1,-1.3)},sibling distance=10mm]  
        child{node[label=below:\lf{1}{2}{\delta-1}]{}}
        child{node[draw=none,fill=white]{...} edge from parent [draw=none]}
        child{node[label=below:\lf{1}{2}{2}]{}}
        child{node[label=below:\lf{1}{2}{1}]{}}
      }      
      child[level distance=30mm,sibling distance=24mm]{
        node[label=right:\cp{1}{ \delta-1}]{}    edge from parent[dotted,thick] 
        [level distance=15mm,sibling distance=18mm]
        child[solid,thin]{
          node[label=right:\nd{1}{{\delta-1}}]{} 
           [grow via three points={one child at (0,-1.3) and two children at (0,-1.3) and (-1.3,-1.3)},sibling distance=10mm]
          child{node[label=below:\lf{1}{\delta-1}{2}]{}}
          child{node[label=below:\lf{1}{\delta-1}{1}]{}}
        }
        child[thin,solid]{
          node[label=right:\fn{1}{\delta}]{} edge from parent[solid]
          [sibling distance = 10mm, level distance = 13mm]
          child{
            node[label=right:\fn{1}{\delta+1}]{}
          }
        }
      }
    }
    edge (root)
  --   (-20mm,-15mm) 
    node[label=left:\cp{2}{1}]{}
    [sibling distance=12mm, level distance=15mm]
    child{
      node[label=right:\nd{2}{1}]{} 
      [grow via three points={one child at (0,-1.3) and two children at (0,-1.3) and (-1,-1.3)}]
      child{node[label=below:\lf{2}{1}{\delta}]{}}
      child{node[label=below:\lf{2}{1}{\delta-1}]{}}
      child{node[draw=none, fill=white]{...} edge from parent [draw=none]}
      child{node[label=below:\lf{2}{1}{2}]{}}
      child{node[label=below:\lf{2}{1}{1}]{}}
    }
    child[level distance=30mm,sibling distance=24mm]{ 
      node[label=right:\cp{2}{2}]{}
      [sibling distance=12mm, level distance=15mm]
      child{
        node[label=right:\nd{2}{2}]{} 
        [grow via three points={one child at (0,-1.3) and two children at (0,-1.3) and (-1,-1.3)},sibling distance=10mm]  
        child{node[label=below:\lf{2}{2}{\delta-1}]{}}
        child{node[draw=none,fill=white]{...} edge from parent [draw=none]}
        child{node[label=below:\lf{2}{2}{2}]{}}
        child{node[label=below:\lf{2}{2}{1}]{}}
      }      
      child[level distance=30mm,sibling distance=24mm]{
        node[label=right:\cp{2}{ \delta-1}]{}    edge from parent[dotted,thick] 
        [level distance=15mm,sibling distance=18mm]
        child[solid,thin]{
          node[label=right:\nd{2}{{\delta-1}}]{} 
          [grow via three points={one child at (0,-1.3) and two children at (0,-1.3) and (-1.3,-1.3)},sibling distance=10mm]
          child{node[label=below:\lf{2}{\delta-1}{2}]{}}
          child{node[label=below:\lf{2}{\delta-1}{1}]{}}
        }
        child[thin,solid]{
          node[label=right:\fn{2}{\delta}]{} edge from parent[solid]
          [sibling distance = 10mm, level distance = 13mm]
          child{
            node[label=right:\fn{2}{\delta+1}]{}
          }
        }
      }
    }
edge (root)
  --   (20mm,-15mm) 
 node[draw=none, fill=white]{\Huge \bfseries \ldots } 
  -- (80mm,-15mm) 
    node[label=right:\cp{\nwidgets}{1}]{}
    [sibling distance=12mm, level distance=15mm]
    child{
      node[label=right:\nd{\nwidgets}{1}]{} 
      [grow via three points={one child at (0,-1.3) and two children at (0,-1.3) and (-1,-1.3)}]
      child{node[label=below:\lf{\nwidgets}{1}{\delta}]{}}
      child{node[label=below:\lf{\nwidgets}{1}{\delta-1}]{}}
      child{node[draw=none, fill=white]{...} edge from parent [draw=none]}
      child{node[label=below:\lf{\nwidgets}{1}{2}]{}}
      child{node[label=below:\lf{\nwidgets}{1}{1}]{}}
    }
    child[level distance=30mm,sibling distance=24mm]{ 
      node[label=right:\cp{\nwidgets}{2}]{}
      [sibling distance=12mm, level distance=15mm]
      child{
        node[label=right:\nd{\nwidgets}{2}]{} 
        [grow via three points={one child at (0,-1.3) and two children at (0,-1.3) and (-1,-1.3)},sibling distance=10mm]  
        child{node[label=below:\lf{\nwidgets}{2}{\delta-1}]{}}
        child{node[draw=none,fill=white]{...} edge from parent [draw=none]}
        child{node[label=below:\lf{\nwidgets}{2}{2}]{}}
        child{node[label=below:\lf{\nwidgets}{2}{1}]{}}
      }      
      child[level distance=30mm,sibling distance=24mm]{
        node[label=right:\cp{\nwidgets}{ \delta-1}]{}    edge from parent[dotted,thick] 
        [level distance=15mm,sibling distance=18mm]
        child[solid,thin]{
          node[label=right:\nd{\nwidgets}{{\delta-1}}]{} 
          [grow via three points={one child at (0,-1.3) and two children at (0,-1.3) and (-1.3,-1.3)},sibling distance=10mm]
          child{node[label=below:\lf{\nwidgets}{\delta-1}{2}]{}}
          child{node[label=below:\lf{\nwidgets}{\delta-1}{1}]{}}
        }
        child[thin,solid]{
          node[label=right:\fn{\nwidgets}{\delta}]{} edge from parent[solid]
          [sibling distance = 10mm, level distance = 13mm]
          child{
            node[label=right:\fn{\nwidgets}{\delta+1}]{}
            }
          }
        }
      }
edge (root);
\end{tikzpicture}

%% file: RR-8082.bbl
\begin{thebibliography}{10}

\bibitem{agulloPP12}
Emmanuel Agullo, Patrick Amestoy, Alfredo Buttari, Abdou Guermouche, Jean-Yves
  L'Excellent, and Fran\c{c}ois-Henry Rouet.
\newblock Robust memory-aware mappings for parallel multifrontal
  factorizations, 2012.
\newblock {SIAM} conf. on Parallel Processing for Scientific Computing (PP12).

\bibitem{GareyJohnson}
M.~R. Garey and D.~S. Johnson.
\newblock {\em Computers and Intractability, {A} Guide to the Theory of
  {NP}-Completeness}.
\newblock W.H. Freeman and Co, 1979.

\bibitem{gilbert80}
John~R. Gilbert, Thomas Lengauer, and Robert~Endre Tarjan.
\newblock The pebbling problem is complete in polynomial space.
\newblock {\em SIAM J. Comput.}, 9(3), 1980.

\bibitem{gimt:98}
John~R. Gilbert, Gary~L. Miller, and Shang-Hua Teng.
\newblock Geometric mesh partitioning: Implementation and experiments.
\newblock {\em SIAM Journal on Scientific Computing}, 19(6):2091--2110, 1998.

\bibitem{GrahamList}
R.~L. Graham.
\newblock Bounds for certain multiprocessing anomalies.
\newblock {\em Bell System Technical Journal}, XLV(9):1563--1581, 1966.

\bibitem{guermouche04ipdps}
A.~Guermouche and J.-Y. L'Excellent.
\newblock Memory-based scheduling for a parallel multifrontal solver.
\newblock In {\em IPDPS'04}, page~71, 2004.

\bibitem{Hu61}
T.C. Hu.
\newblock Parallel sequencing and assembly line problems.
\newblock {\em Operations Research}, 9, 1961.

\bibitem{Hwang1989:spg}
J.~J. Hwang, Y.~C. Chow, F.~D. Anger, and C.~Y. Lee.
\newblock Scheduling precedence graphs in systems with interprocessor
  communication times.
\newblock {\em SIAM Journal of Computing}, 18(2), 1989.

\bibitem{ipdps-tree-traversal}
Mathias Jacquelin, Loris Marchal, Yves Robert, and Bora Ucar.
\newblock On optimal tree traversals for sparse matrix factorization.
\newblock {\em IPDPS'11}, 2011.

\bibitem{kaku:98:metis}
G.~Karypis and V.~Kumar.
\newblock {\em {MeTiS}: {A} Software Package for Partitioning Unstructured
  Graphs, Partitioning Meshes, and Computing Fill-Reducing Orderings of Sparse
  Matrices Version 4.0}.
\newblock U. of Minnesota, Dpt. of Comp. Sci. and Eng., Army HPC Research
  Center, Minneapolis, 1998.

\bibitem{rauber11CLSS}
Chi-Chung Lam, Thomas Rauber, Gerald Baumgartner, Daniel Cociorva, and
  P.~Sadayappan.
\newblock Memory-optimal evaluation of expression trees involving large
  objects.
\newblock {\em Computer Languages, Systems {\&} Structures}, 37(2):63--75,
  2011.

\bibitem{LenstraRKBr77}
J.~K. Lenstra, A.~H.~G. Rinnooy~Kan, and P.~Brucker.
\newblock Complexity of machine scheduling problems.
\newblock {\em Annals of Discrete Mathematics}, 1:343--362, 1977.

\bibitem{Liu86}
Joseph W.~H. Liu.
\newblock On the storage requirement in the out-of-core multifrontal method for
  sparse factorization.
\newblock {\em ACM Trans. Math. Software}, 12(3):249--264, 1986.

\bibitem{Liu87}
Joseph W.~H. Liu.
\newblock An application of generalized tree pebbling to sparse matrix
  factorization.
\newblock {\em SIAM J. Algebraic Discrete Methods}, 8(3), 1987.

\bibitem{ramakrishnan07ccgrid}
Arun Ramakrishnan, Gurmeet Singh, Henan Zhao, Ewa Deelman, Rizos Sakellariou,
  Karan Vahi, Kent Blackburn, David Meyers, and Michael Samidi.
\newblock Scheduling data-intensiveworkflows onto storage-constrained
  distributed resources.
\newblock In {\em {CCGRID'07}}. IEEE, 2007.

\bibitem{Sethi73}
Ravi Sethi.
\newblock Complete register allocation problems.
\newblock In {\em {STOC'73}}, pages 182--195. ACM Press, 1973.

\bibitem{SethiUllman70}
Ravi Sethi and J.D. Ullman.
\newblock The generation of optimal code for arithmetic expressions.
\newblock {\em J. ACM}, 17(4):715--728, 1970.

\end{thebibliography}
